\documentclass[letterpaper, 10 pt, journal, twoside]{IEEEtran}
\usepackage[numbers,compress]{natbib}

\usepackage{amsmath, amssymb,mathrsfs, dsfont} 

\usepackage{bm,bbm}
\DeclareMathAlphabet{\mathbbold}{U}{bbold}{m}{n}

\usepackage{xspace}

\usepackage{tikzsymbols}
\usepackage{color}
\usepackage{graphicx,graphics} 
\usepackage{epstopdf} 
\usepackage{caption}
\usepackage{algorithm, algorithmicx, multirow, titlecaps} %
\usepackage{threeparttable}
\usepackage[english]{babel}
\usepackage[noend]{algpseudocode} 
\usepackage{tabularx}
\usepackage{booktabs}
\usepackage{multirow}
\usepackage{courier}
\usepackage{float}
\usepackage{makecell}
\setcellgapes{5pt}
\newcommand{\ra}[1]{\renewcommand{\arraystretch}{#1}}
\usepackage{array}
\usepackage{adjustbox}

\usepackage{authblk}
\usepackage{url} 

\newcommand{\smtiny}[1]{{\scalebox{.63}{#1}}}
\newcommand{\stiny}[1]{{\scalebox{.5}{#1}}}

\newcommand*{\SumOp}{\operatornamewithlimits{\text{\scalebox{1.25}{$\sum$}}}\limits}

\newcommand*{\limOp}{\operatornamewithlimits{lim}\limits}

\newcommand{\tr}{{\smtiny{$\mathsf{T}$ }}\!}

\newcommand{\zero}{\mathbf{0}}

\iftrue
\newcommand{\eye}{\mathbb{I}}
\else
\newcommand{\eye}{\mathbf{I}}
\fi

\newcommand{\vc}[1]{{ \mathrm{#1} }}
\newcommand{\mx}[1]{{ \mathrm{#1} }}

\newcommand{\drm}{\mathrm{d}}

\newcommand{\inner}[2]{{ \langle {#1,#2} \rangle}}

\newcommand{\ellone}{\ell^{1}}
\newcommand{\ellinfty}{\ell^{\infty}}

\newcommand{\Lone}{L^{1}}
\newcommand{\Linfty}{L^{\infty}}

\newcommand{\Lscrone}{\Lscr^{1}}
\newcommand{\Lscrinfty}{\Lscr^{\infty}}
\newcommand{\Lscrp}{\Lscr^{p}}

\newcommand{\Dscr}{{\mathscr{D}}}

\newcommand{\Hscr}{{\mathscr{H}}}

\newcommand{\Lscr}{{\mathscr{L}}}

\newcommand{\Xscr}{{\mathscr{X}}}

\newcommand{\Jcal}{{\mathcal{J}}}

\newcommand{\Scal}{{\mathcal{S}}}

\newcommand{\Nbb}{{\mathbb{N}}}

\newcommand{\Rbb}{{\mathbb{R}}}

\newcommand{\Tbb}{{\mathbb{T}}}

\newcommand{\Zbb}{{\mathbb{Z}}}

\newcommand{\bbk}{\mathbbm{k}}

\usepackage{amsthm}

\newtheorem{theorem}{Theorem}
\newtheorem*{theorem*}{Theorem}
\newtheorem{definition}{Definition}
\newtheorem*{definition*}{Definition}

\newtheorem{corollary}[theorem]{Corollary}
\newtheorem{lemma}[theorem]{Lemma}

\newtheorem*{example*}{Example}

\newtheorem*{claim*}{Claim}

\newtheorem*{problem*}{Problem}

\iftrue 
\usepackage[colorlinks=true]{hyperref}
\usepackage{xcolor}
\hypersetup{
	colorlinks,
	linkcolor={blue!55!black},
	citecolor={red!55!black},
	urlcolor={blue!80!black}
}
\allowdisplaybreaks

\fi

\newcommand{\utau}{\underline{\tau}}
\newcommand{\otau}{\overline{\tau}}

\newcommand{\tin}{t_i^{(n)}}

\newcommand{\tjm}{t_j^{(m)}}
\newcommand{\Iin}{\mathrm{I}_i^{(n)}}

\newcommand{\Iijnm}{\mathrm{I}_{i,j}^{(n,m)}}

\newcommand{\Jin}{\mathrm{J}_i^{(n)}}

\newcommand{\Jjm}{\mathrm{J}_j^{(m)}}

\newcommand{\nD}{{n_{\stiny{$\!\Dscr$}}}}

\newcommand{\Lu}[1]{\mx{L}^{\!\vc{u}}_{#1}}
\newcommand{\gS}{\vc{g}^{\stiny{$(\Scal)$}}}
\newcommand{\gtS}{{g}^{\stiny{$(\Scal)$}}}

\newcommand{\gstar}{\vc{g}^{\star}}

\newcommand{\vcf}{\vc{f}}
\newcommand{\vcg}{\vc{g}}
\newcommand{\vch}{\vc{h}}

\newcommand{\vcu}{\vc{u}}
\newcommand{\vcv}{\vc{v}}

\newcommand{\vcy}{\vc{y}}

\newcommand{\mxO}{\mx{O}}

\newcommand{\vcc}{\vc{c}}

\newcommand{\kernel}{\bbk}

\newcommand{\Hk}{\Hscr_{\bbk}}

\newcommand{\Dtau}{\Delta\tau}

\newcommand{\phiu}[1]{\varphi_{#1}^{\text{\rm{(u)}}}}

\newcommand{\vcvn}{\vc{v}^{(n)}}

\newcommand{\vn}{v^{(n)}}
\newcommand{\vm}{v^{(m)}}
\newcommand{\an}{a^{(n)}}
\newcommand{\am}{a^{(m)}}

\usepackage{multirow}
\usepackage{booktabs}
\usepackage{array}
\usepackage{lineno}
\title{\textbf{\LARGE The Existence and Uniqueness of Solutions for Kernel-Based System Identification}}
\author[$\dagger$]{Mohammad Khosravi}%
\author[$\dagger$]{Roy S. Smith}
\affil[$\dagger$]{Automatic Control Laboratory, ETH Z\"urich 
	\authorcr
	\texttt{\{khosravm,rsmith\}@control.ee.ethz.ch}
}
\begin{document}
	\maketitle                           
\begin{abstract}
The notion of reproducing kernel Hilbert space (RKHS) has emerged in system identification during the past decade. In the resulting framework, the impulse response estimation problem is formulated as a regularized optimization defined on an infinite-dimensional RKHS consisting of stable impulse responses. The consequent estimation problem is well-defined under the central assumption that the convolution operators restricted to the RKHS are continuous linear functionals. Moreover, according to this assumption, the representer theorem hold, and therefore, the impulse response can be estimated by solving a finite-dimensional program. Thus, the continuity feature plays a significant role in kernel-based system identification. This paper shows that this central assumption is guaranteed to be satisfied in considerably general situations, namely when the kernel is an integrable function and the input signal is bounded. Furthermore, the strong convexity of the optimization problem and the continuity property of the convolution operators imply that the kernel-based system identification admits a unique solution. Consequently, it follows that kernel-based system identification is a well-defined approach.
\end{abstract}

\section{Introduction}\label{sec:introduction} \vspace{0mm}
System identification, the theory of generating suitable abstract representations for dynamical systems based on measurement data, is a well-established research field \cite{zadeh1956identification}. Due to the importance of mathematical models in various areas of science and technology, system identification is an active research area with numerous developed methodologies \cite{ljung2010perspectives,LjungBooK2, schoukens2019nonlinear,khosravi2021ROA,ahmadi2020learning,khosravi2021grad}.
On the other hand, the concept of reproducing kernel Hilbert space (RKHS), initially introduced in \cite{aronszajn1950theory}, has emerged in statistics, signal processing and numerical analysis \cite{parzen1959statistical,parzen1961approach,kailath1971rkhs,wahba1990spline,cucker2002best,berlinet2011reproducing}, and provided a solid foundation for estimation and interpolation problems. The inherent features of RKHSs, such as their fundamental relation to the positive semi-definite kernels and the Gaussian process \cite{kimeldorf1970correspondence,lukic2001stochastic,kanagawa2018gaussian}, led to establishing various methodologies and opened numerous avenues of research in statistical learning theory  \cite{cuckerANDsmale2002mathematical}.  \vspace{0mm}

In the seminal work of Pillonetto and De Nicolao \cite{pillonetto2010new}, the \emph{kernel-based identification} methods are introduced by bringing the theory of RKHSs to the area of linear system identification, which led to a paradigm shift in the field \cite{ljung2020shift}.
The kernel-based method unifies the identification theory of continuous-time systems and discrete-time systems, described either with a finite or an infinite impulse response, by formulating the identification problem as a regularized regression defined on a RKHS of stable systems, where the regularization term is specified based on the norm of employed RKHS \cite{pillonetto2014kernel}. The resulting formulation addresses issues of  model order selection, robustness, and bias-variance trade-off \cite{pillonetto2014kernel,chiuso2019system,khosravi2021robust}. 
The cornerstone of a RKHS is the associated kernel function, which highlights the necessity of designing suitable kernels for system identification \cite{dinuzzo2015kernels}. The most frequently used kernels in the literature are {tuned/correlated} (TC),  {diagonal/correlated} (DC), {stable spline} (SS), and their generalizations \cite{zorzi2021second,chen2018continuous,andersen2020smoothing}. 
Other forms of kernels and regularization matrices have been proposed, inspired by machine learning, system theory, harmonic analysis of stochastic processes, and filter design methods \cite{chen2014system,zorzi2018harmonic,marconato2016filter}.
While in the classical identification methods, the complexity of models is described by the orders of system, which are integer variables determined based on metrics such as Akaike information criterion \cite{LjungBooK2}, the model complexity in kernel-based approach is specified and regulated by the hyperparameters characterizing the kernel and the regularization weight, which are continuous variables to be tuned \cite{ljung2020shift}.
The estimation of hyperparameters can be performed using powerful and robust methods such as empirical Bayes, Stein unbiased risk estimator, and cross-validation \cite{pillonetto2015tuning,mu2018asymptotic,mu2018asymptotic-GCV,mu2021asymptotic}. 
Moreover, the kernel-based scheme allows the incorporation of various forms of side-information in the identification problem by designing appropriate kernel functions or imposing suitable constraints to the regression problem.
The forms of this side-information, studied to date, include  stability, relative degree,  smoothness of the impulse response,  resonant frequencies,  external positivity, oscillatory behaviors, steady-state gain, internal  positivity, exponential decay of the impulse response, structural properties, internal low-complexity, frequency domain features, and the presence of fast and slow poles \cite{fujimoto2017extension,marconato2016filter,chen2012estimation,zheng2021bayesian,pillonetto2016AtomicNuclearKernel,khosravi2020low,darwish2018quest,chen2018kernel,khosravi2019positive,prando2017maximum,fujimoto2018kernel,risuleo2019bayesian,everitt2018empirical,risuleo2017nonparametric,khosravi2020regularized,khosravi2021POS,khosravi2021SSG,khosravi2021FDI}.
While kernel-based system identification has enjoyed considerable progress in the past decade, it is still a thriving area of research with state-of-the-art results and recent studies \cite{scandella2021kernel,pillonetto2019stable,scandella2020note,bisiacco2020mathematical,pillonetto2021sample,bisiacco2020kernel}. For example, the mathematical foundation of stable RKHSs is revisited in  \cite{bisiacco2020mathematical}, the sample complexity and the minimax properties of kernel-based methods are discussed in \cite{pillonetto2021sample}, and a long-standing question on the absolute summability of stable kernels is addressed in \cite{bisiacco2020kernel}.  \vspace{0mm}

The above-mentioned advantages of kernel-based methods stand on the assumption that the formulated regression problem is well-defined, i.e., the corresponding regularized optimization problem admits at least one solution. The base of this assumption is the continuity of convolution operators when they are restricted to the stable RKHS \cite{pillonetto2014kernel,dinuzzo2015kernels}. Accordingly, one may ask about the  conditions under which the continuity property holds.  This paper shows that this central assumption is satisfied in certain but highly general situations, namely when the input signal is bounded and the kernel is integrable. As a result,  kernel-based system identification admits a unique solution according to the continuity of convolution operators and the strong convexity of the optimization problem, which also implies that the kernel-based approach is well-defined.  \vspace{0mm}

\section{Notation and Preliminaries} 
The set of natural numbers, the set of non-negative integers, the set of real numbers, the set of non-negative real numbers, and the $n$-dimensional Euclidean space are denoted by $\Nbb$, $\Zbb_+$,  $\Rbb$, $\Rbb_+$, and $\Rbb^n$, respectively.
Throughout the paper, $\Tbb$ denotes either $\Zbb_+$ or $\Rbb_+$, and $\Tbb_{\pm}$ is defined as the set of scalars $t$ where $t\in\Tbb$ or $-t\in\Tbb$.  
The identity matrix/operator and the zero vector are denoted by $\eye$ and $\zero$, respectively.
Given measurable space $\Xscr$, we denote by $\Rbb^{\Xscr}$ as the space of measurable functions $\vcv:\Xscr\to \Rbb$.
The element $\vcv\in\Rbb^{\Xscr}$ is shown entry-wise as $\vcv=(v_x)_{x\in\Xscr}$, or $\vcv=(v(x))_{x\in\Xscr}$.
Depending on the context of discussion, $\Lscrinfty$ refers either to $\ellinfty(\Zbb)$ or $\Linfty(\Rbb)$. 
Similarly, $\Lscrone$ is either $\ellone(\Zbb_+)$ or $\Lone(\Rbb_+)$. 
For $p\in\{1,\infty\}$, the norm in $\Lscrp$ is denoted by $\|\cdot\|_{p}$.
With respect to each $\vcu=(u_s)_{s\in \Tbb_{\pm}}\in \Lscrinfty$ and $t\in\Tbb_{\pm}$, the linear operator $\Lu{t}:\Lscrone\to\Rbb$ is defined as
$\Lu{t}(\vcg) := \sum_{s\in \Zbb_+}g_s u_{t-s}$, when $\Tbb=\Zbb_+$, and 
$\Lu{t}(\vcg) := \int_{\Rbb_+}\!\!g_s u_{t-s}\drm s$, when $\Tbb=\Zbb_+$.

\section{Kernel-Based System Identification} \label{sec:problem_statement}   
Consider a stable LTI system $\Scal$ characterized by an impulse response  $\gS:=(\gtS_t)_{t\in\Tbb}\in\Rbb^{\Tbb}$, 
where $\Tbb$ is $\Zbb_+$ or $\Rbb_+$ respectively for the case that the system is discrete-time or continuous-time.
Suppose the system $\Scal$ is actuated by a signal  $\vcu\in\Lscrinfty$, and the resulting output signal is measured with measurement noise at 
$\nD$ time instants $t_1,\ldots,t_{\nD}$.
Let the measured output of the system at time instant $t_i$, and the corresponding measurement uncertainty, be denoted by $y_{t_i}$ and $w_{t_i}$, respectively. Due to the definition of operators $\{\Lu{t}|t\in\Tbb_{\pm}\}$, we know that \vspace{0mm}
\begin{equation}\label{eqn:output_sys_S} \vspace{0mm}
	y_{t_i} = \Lu{t_i}(\gS)+w_{t_i}, \qquad i=1,\ldots,\nD.
\end{equation}
Therefore, we are provided with a set of input-output measurement data denoted by $\Dscr$.
Accordingly, the impulse response identification problem is formalized as estimating $\gS$, the impulse response of stable system $\Scal$, based on the measurement data. In the kernel-based  identification framework, this problem is formulated as an impulse response estimation in a reproducing kernel Hilbert space (RKHS) endowed with a stable kernel. 
To introduce the main result of this paper, we need to discuss this paradigm briefly. \vspace{0mm}
\begin{definition}[\cite{berlinet2011reproducing}] \vspace{0mm}
	\label{def:kernel_and_section}
	Consider symmetric function $\kernel:\Tbb\times\Tbb\to \Rbb$, that is assumed to be continuous if $\Tbb=\Rbb_+$. We say $\kernel$ is a {\em Mercer kernel} when we have \vspace{0mm}
	\begin{equation} \vspace{0mm}
		\sum_{i=1}^{n}\sum_{j=1}^{n}
		a_i\kernel(t_i,t_j)a_j\ge 0,
	\end{equation} 
	for all $n\in\Nbb$, $t_1,\ldots,t_n\in\Tbb$, and $a_1,\ldots,a_n\in\Rbb$.
	Furthermore, with respect to each $t\in\Tbb$, the {\em section} of kernel $\kernel$ at $t$ is the function $\kernel_t:\Tbb\to\Rbb$ defined as $\kernel_{t}(\cdot)=\kernel(t,\cdot)$.
\end{definition} \vspace{0mm}
\begin{theorem}[\cite{berlinet2011reproducing}]\label{thm:kernel_to_RKHS_def} \vspace{0mm}
	With respect to each Mercer kernel $\kernel:\Tbb\times\Tbb\to \Rbb$, a unique Hilbert space $\Hk\subseteq \Rbb^{\Tbb}$ endowed with inner product $\inner{\cdot}{\cdot}_{\Hk}$ exists such that, for each $t\in\Tbb$, one has\\
	\emph{i)} $ \kernel_t\in\Hk$, and\\
	\emph{ii)} $\inner{\vc{g}}{ \kernel_{t}}_{\Hk}=g_t$, for all $\vc{g}=(g_s)_{s\in\Tbb}\in\Hk$.\\
	In this case, we say $\Hk$ is the \emph{RKHS with kernel} $\kernel$. Moreover, the second feature is called the {\em reproducing property}.
\end{theorem} \vspace{0mm}
Due to Theorem~\ref{thm:kernel_to_RKHS_def}, one can see that each RKHS is uniquely characterized by the corresponding Mercer kernel.
Since the to-be-estimated impulse response is known to be stable in the bounded-input-bounded-output (BIBO) sense, the employed kernel $\kernel$ is required to guarantee that $\Hk\subseteq\Lscr^1$.
The sufficient and necessary condition for this property is established by the following theorem.  \vspace{0mm}
\begin{theorem}[\cite{chen2018stability,carmeli2006vector}]
	Consider the Mercer kernel $\kernel:\Tbb\times\Tbb\to \Rbb$ and the corresponding RKHS $\Hk$. Then, $\Hk\subseteq\Lscr^1$  if and only if, for any $\vc{u}=(u_s)_{s\in\Tbb}\in\Lscr^{\infty}$,  
	one has \vspace{0mm}
	\begin{equation} \vspace{0mm}
		\sum_{t\in\Zbb_+}\bigg|\sum_{s\in\Zbb_+}u_s\kernel(t,s)\bigg|<\infty,
	\end{equation} 
	when $\Tbb=\Zbb_+$, and, \vspace{0mm}
	\begin{equation} \vspace{0mm}
		\int_{\Rbb_+}\bigg|\int_{\Rbb_+}u_s\kernel(t,s)\drm s\bigg|\drm t<\infty,
	\end{equation} 
	when $\Tbb=\Rbb_+$.
	When this property holds, kernel $\kernel$ is called \emph{stable} and $\Hk$ is said to be a \emph{stable RKHS}.
\end{theorem} \vspace{0mm}
Given the stable kernel $\kernel$ and the measurement data, the \emph{kernel-based impulse response estimation} problem is formulated as  \vspace{0mm}
\begin{equation}\label{eqn:kernel_based_IR_identification} \vspace{0mm}
	\min_{\vcg\in\Hk}\ \sum_{i=1}^{\nD}(\Lu{t_i}(\vcg)-y_i)^2 + \lambda \|\vcg\|_{\Hk}^2,
\end{equation}
where $\lambda>0$ is the regularization weight.
Based on the same arguments as in \cite[Theorem 1.3.1]{wahba1990spline}, one can describe the solution of \eqref{eqn:kernel_based_IR_identification} in terms of the sections of the kernel at $t_1,\ldots,t_{\nD}$. 
To this end, we need vector $\vcy$ defined as $\vcy=\big[y_{t_1},\ldots,y_{t_{\nD}}\big]^\tr\in\Rbb^{\nD}$, and the \emph{output kernel matrix} $\mxO\in\Rbb^{\nD\times\nD}$ formed from the input signal and defined entry-wise as  \vspace{0mm}
\begin{equation*} \vspace{0mm}
[\mxO]_{(i,j)}\!=\!\left\{
\!
\ra{1.8}
\begin{array}{ll}
	\int_{\Rbb_+}\!\int_{\Rbb_+}\!\kernel(s,t)u_{t_{i}-s}u_{t_{j}-t}\!\ \drm s \!\ \drm t, &\ \text{ if } \Tbb_+\!=\!\Rbb_+,\\
	\SumOp_{t\in\Zbb_+}\SumOp_{s\in\Zbb_+}\!\kernel(s,t)u_{t_{i}-s}u_{t_{j}-t},   &\ \text{ if } \Tbb_+\!=\!\Zbb_+,\\
\end{array}\right.
\end{equation*} 
for each $i,j=1,\ldots,\nD$. \vspace{0mm}
\begin{theorem}[Representer Theorem for System Identification, \cite{pillonetto2014kernel}]\label{thm:kernel_based_IR_identification}
Let $\Lu{t_i}:\Hk\to\Rbb$ be a continuous linear operator, for each $i=1,\ldots,\nD$. Then, the minimizer of \eqref{eqn:kernel_based_IR_identification} is $\gstar=(g_t^{\star})_{t\in\Tbb}\in\Hk$ defined as  \vspace{0mm}
\begin{equation}\label{eqn:solution_KRI} \vspace{0mm}
g_t^{\star} = \sum_{i=1}^{\nD}c_i^\star\Lu{t_i}(\kernel_t),\qquad \forall t\in\Tbb,
\end{equation}
where the vector $\vcc^\star=\big[c_1^\star,\ldots,c_{\nD}^\star\big]^\tr\in\Rbb^{\nD}$ is  
\begin{equation} \vspace{0mm}
	\vcc^\star = \big(\mxO+\lambda\eye_{\nD}\big)^{-1}\vcy,
\end{equation}
and, $\eye_{\nD}$ denotes  identity matrix of dimension $\nD$.
\end{theorem}
The main assumption in Theorem~\ref{thm:kernel_based_IR_identification} is the continuity of convolution operators $\Lu{t_1},\ldots,\Lu{t_{\nD}}$, which depends mainly on the input signal $\vcu$ and kernel $\kernel$.
Accordingly, a natural question one may ask is \emph{under what conditions are the convolution operators continuous}.
Indeed, one should note that in Theorem~\ref{thm:kernel_based_IR_identification}, the convolution operators are restricted to $\Hk\subset\Lscrone$, and consequently, the continuity of $\Lu{t}:\Lscrone\to\Rbb$ does not imply that the restricted operator $\Lu{t}:\Hk\to\Rbb$ is continuous as well.
We address this continuity concern in the next section.

\section{Continuity of Convolution Operators}\label{sec:cont_Lu} 
The main result of this section is based on the notion of integrable kernels introduced below. 
\begin{definition}[\cite{pillonetto2014kernel}]\label{def:integrable_kernel}
	The Mercer kernel $\kernel:\Tbb\times\Tbb\to \Rbb$ is said to be \emph{integrable}
	if 
	\begin{equation} 
		\int_{\Rbb_+}\int_{\Rbb_+}|\kernel(s,t)|\ \! \drm s \drm t<\infty,
	\end{equation}
	when $\Tbb=\Rbb_+$, or, if  
	\begin{equation}\label{eqn:abs_summable} 
		\sum_{s\in\Zbb_+}
		\sum_{t\in\Zbb_+}
		|\kernel(s,t)|<\infty,
	\end{equation} 
	when $\Tbb=\Zbb_+$.
\end{definition} 
The integrable kernels are the largest known interesting sub-class of stable kernels in the context of kernel-based impulse response identification \cite{bisiacco2020kernel,bisiacco2020mathematical}. 
Before proceeding to the main theorem of this paper, we need to introduce additional lemmas. 
\begin{lemma}\label{lem:int_kernel_I}
	Let $\Tbb=\Rbb_+$ and kernel $\kernel$ be integrable. Consider $\utau$ and $\otau$ such that $0\le \utau<\otau\le \infty$.
	Then, $\int_{[\utau,\otau]}\kernel(\cdot,t)\drm t$ is a well-defined function and belongs to $\Hk$ for which we have 
	\begin{equation} 
		\Big\|\int_{[\utau,\otau]}\kernel(\cdot,t)\drm t\Big\|_{\Hk}^2 =
		\int_{\utau}^{\otau}\int_{\utau}^{\otau}\kernel(s,t)\drm s\drm t.
	\end{equation}
	Moreover, for each $\vcg=(g_t)_{t\in\Rbb_+}\in\Hk$, the following holds 
	\begin{equation}\label{eqn:int_g} 
		\int_{[\utau,\otau]}g_t\drm t =
		\Big\langle\int_{[\utau,\otau]}\kernel(\cdot,t)\drm t,\vcg\Big\rangle_{\Hk}.
	\end{equation}
\end{lemma} 
\begin{proof} 
	See Appendix \ref{apn:proof:lem:int_kernel_I}.
\end{proof} 
From this Lemma, we have the following corollary. 
\begin{corollary}
	\label{cor:inner_int_I1_int_I2}
	Let $\Tbb=\Rbb_+$ and kernel $\kernel$ be integrable. Consider $\utau_1$, $\utau_2$, $\otau_2$ and $\otau_2$ such that $0\le \utau_1<\otau_1\le \infty$ and $0\le \utau_2<\otau_2\le \infty$.
	Then, we have 
	\begin{equation}\label{eqn:inner_int_I1_int_I2} 
		\begin{split}
			&\!\!\!
			\Big\langle\!\int_{[\utau_1,\otau_1]}
			\!\! \kernel(\cdot,t)\drm t,\int_{[\utau_2,\otau_2]}
			\!\!\kernel(\cdot,s)\drm s\Big\rangle_{\Hk}
			\\&
			=
			\int_{[\utau_1,\otau_1]\times[\utau_2,\otau_2]}	
			\kernel(t,s)\drm s \drm t
			\\&
			=
			\int_{\utau_1}^{\otau_1}\int_{\utau_2}^{\otau_2}
			\kernel(t,s)\drm s \drm t
			=
			\int_{\utau_2}^{\otau_2}\int_{\utau_1}^{\otau_1}
			\kernel(t,s)\drm t \drm s.
		\end{split}
	\end{equation}
\end{corollary} 
\begin{proof} 
	See Appendix \ref{apn:proof:cor:inner_int_I1_int_I2}.
\end{proof} 
The next lemma is the discrete-time version of Lemma~\ref{lem:int_kernel_I}. 
\begin{lemma}\label{lem:sum_kernel_I}
	Let $\Tbb=\Zbb_+$ and kernel $\kernel$ be integrable. Consider $\utau,\otau\in\Zbb_+$ such that $0\le \utau\le \otau\le \infty$.
	Then, $\sum_{\utau\le t\le\otau}\kernel(\cdot,t)$ is a well-defined function and belongs to $\Hk$ for which we have 
	\begin{equation}\label{eqn:sum_kt} 
		\Big\|\sum_{\utau\le t\le\otau}\kernel(\cdot,t)\Big\|_{\Hk}^2 =
		\sum_{\utau\le s,t\le\otau}\kernel(s,t).
	\end{equation}
	Moreover, for each $\vcg=(g_t)_{t\in\Zbb_+}\in\Hk$, the following holds 
	\begin{equation}\label{eqn:sum_gt} 
		\sum_{\utau\le t\le\otau}g_t =
		\Big\langle\sum_{\utau\le t\le\otau}\kernel(\cdot,t)\drm t,\vcg\Big\rangle_{\Hk}.
	\end{equation}
\end{lemma}
\begin{proof} 
	See Appendix \ref{apn:proof:lem:sum_kernel_I}.
\end{proof}
Based on Lemma \ref{lem:int_kernel_I}, Corollary \ref{cor:inner_int_I1_int_I2} and Lemma \ref{lem:sum_kernel_I}, we can present the main theorem of this paper. 
\begin{theorem}[Continuity of Convolution Operators]\label{thm:Lu_bounded} 
	Let $\kernel$ be an integrable Mercer kernel.
	Then, for any $\vc{u}\in\Lscrinfty$ and $\tau\in\Tbb$, the operator $\Lu{\tau}$ is continuous (bounded). 
	Moreover, 
	there exists $\phiu{\tau}=(\phiu{\tau,t})_{t\in\Tbb}\in\Hk$ such that $\Lu{\tau}(\vcg) = \inner{\phiu{\tau}}{\vcg}$, for any $\vcg\in\Hk$.
	Furthermore, for any $t\in\Tbb$, 
	we have  
	\begin{equation}\label{eqn:phiu_tau_t} 
		\phiu{\tau,t} =  \Lu{\tau}(\kernel_t) = 
		\begin{cases}
			\int_{\Rbb_+}\kernel(t,s)u_{\tau-s}\drm s,
			&
			\text{ if } \Tbb=\Rbb_+,
			\\
			\sum_{s\in\Zbb_+}\kernel(t,s)u_{\tau-s},
			&
			\text{ if } \Tbb=\Zbb_+.
		\end{cases}
	\end{equation}
\end{theorem}
\begin{proof}
	We discuss the proof for the cases of $\Tbb=\Rbb_+$ and $\Tbb=\Zbb_+$.\\
	\textbf{Case I:} Let $\Tbb=\Rbb_+$ and
	define $\vcv=(v_s)_{s\in\Rbb_+}$ such that $v_s=u_{\tau-s}$, for any $s\in\Rbb_+$.
	Accordingly, we have 
	\begin{equation} 
		\Lu{\tau}(\vcg)
		=
		\int_{\Rbb_+}v_s g_s \drm s,
	\end{equation}
	for each $\vcg=(g_s)_{s\in\Rbb_+}$.
	Note that $\vcv\in\Lscrinfty$, and hence, in $\Lscrinfty$, there exists a sequence of step functions $\vcvn:=(\vn_s)_{s\in\Rbb_+}$, $n=1,2,\ldots,$  such that $\|\vcvn\|_{\infty}\le \|\vcv\|_{\infty}$, for each $n\in\Nbb$, and, $\lim_{n\to\infty}\vn_s=v_s$, for almost all $s\in\Rbb_+$ \cite{stein2009real}.
	For each $n\in\Nbb$, due to the definition of step functions \cite{stein2009real}, we know that there exists $M_n\in\Nbb$, intervals $\Jin\subseteq\Rbb_+$, $i=1,\ldots,M_n$, and $\an_i\in\Rbb$, $i=1,\ldots,M_n$, 
	such that 
	\begin{equation}\label{eqn:vn} 
		\vn_s = \sum_{i=1}^{M_n}\an_i\mathbf{1}_{\Jin}(s), \qquad \forall s\in\Rbb_+.
	\end{equation}
	For each $n\in\Nbb$, define $\vcf_n=(f_{n,t})_{t\in\Rbb_+}$ as  
	\begin{equation}\label{eqn:f_n} 
		\vcf_n := \int_{\Rbb_+}\!\!\!\!\vn_s\kernel(\cdot,s)\drm s
		=
		\int_{\Rbb_+}\sum_{i=1}^{M_n}\!\an_i\mathbf{1}_{\Jin}(s)\kernel(\cdot,s)\drm s,
	\end{equation}
	which is well-defined and belongs to $\Hk$ according to Lemma \ref{lem:int_kernel_I}.
	Accordingly, due to \eqref{eqn:vn}, for each $\vcg=(g_s)_{s\in\Rbb_+}\in\Hk$, we have 
	\begin{equation} 
		\label{eqn:int_gs_vns_inner_fn_g}
		\begin{split}
			\int_{\Rbb_+}& g_s\vn_s\drm s 
			=
			\sum_{i=1}^{M_n}\an_i\int_{\Jin}g_s\drm s 
			\\&=
			\sum_{i=1}^{M_n}\an_i\Big\langle\int_{\Jin}\kernel(\cdot,s)\drm s,\vcg\Big\rangle_{\Hk} 
			\\&=
			\Big\langle\sum_{i=1}^{M_n}\an_i\int_{\Jin}\kernel(\cdot,s)\drm s,\vcg\Big\rangle_{\Hk}
			\\&=
			\Big\langle\int_{\Rbb_+}\sum_{i=1}^{M_n}\an_i\mathbf{1}_{\Jin}(s)\kernel(\cdot,s)\drm s,\vcg\Big\rangle_{\Hk} 
			\\&=
			\Big\langle\int_{\Rbb_+}\vn_s\kernel(\cdot,s)\drm s,\vcg\Big\rangle_{\Hk} 
			\\&= 
			\langle\vcf_n,\vcg\rangle_{\Hk},
		\end{split}
	\end{equation}
	where the second equality is due to Lemma \ref{lem:int_kernel_I}.
	Let $\varepsilon$ be an arbitrary positive real scalar.
	Define $l$ as 
	\begin{equation} 
		l:=\int_{\Rbb_+\times \Rbb_+}v_t\kernel(t,s)v_s \drm t \drm s,
	\end{equation}
	and, for any $n\in\Nbb$, $l_n$ as  
	\begin{equation} 
		l_n:=\int_{\Rbb_+\times \Rbb_+}\vn_t\kernel(t,s) v_s \drm t \drm s.
	\end{equation} 
	For almost all $s,t\in\Rbb_+$, we have 
	\begin{equation} 
		\lim_{n\to\infty}\vn_t\kernel(t,s)v_s = 
		v_t\kernel(t,s)v_s. 
	\end{equation}
	Moreover, for any $n\in\Nbb$, we know that  
	\begin{equation} 
		\big|
		\vn_t\kernel(t,s)v_s
		\big|
		\le \|\vcv\|_{\infty} |\kernel(t,s)|.
	\end{equation}
	Since $\kernel$ is integrable, from the dominated convergence theorem~\cite{stein2009real}, we have 
	\begin{equation} 
		\begin{split}
			\lim_{n\to\infty}l_n 
			&=
			\lim_{n\to\infty} \int_{\Rbb_+\times\Rbb_+}
			\vn_t\kernel(t,s)v_s \drm s \drm t
			\\&=
			\int_{\Rbb_+\times\Rbb_+}
			\lim_{n\to\infty} 
			\vn_t\kernel(t,s)v_s \drm s \drm t
			\\&=
			\int_{\Rbb_+\times\Rbb_+}
			v_t\kernel(t,s)v_s \drm s \drm t = l.
		\end{split}
	\end{equation}
	Therefore, there exist $N_{\varepsilon}\in\Nbb$ such that $|l_n-l|\le \frac18\varepsilon^2$, for each $n\ge N_{\varepsilon}$.
	Define $l_{n,m}$ as
	$l_{n,m}:=\inner{\vcf_n}{\vcf_m}_{\Hk} $, for each $m,n\in\Nbb$.
	Accordingly, from \eqref{eqn:f_n}, Corollary \ref{cor:inner_int_I1_int_I2} and the linearity of integration and inner product, it follows that 
	\begin{equation*} 
		\begin{split}
			l_{n,m}
			\!\! \!\! \!\! \!\! \!\! \!\! 
			&
			\ \ \ \ \ \ 
			=
			\sum_{i=1}^{M_n}
			\sum_{j=1}^{M_m}
			\an_i\am_j
			\Big\langle \int_{\Jin}\!\!\kernel(\cdot,s)\drm s,\int_{\Jjm}\!\!\kernel(\cdot,t)\drm t\Big\rangle_{\Hk} 
			\\&=
			\sum_{i=1}^{M_n}
			\sum_{j=1}^{M_m}
			\an_i\am_j
			\int_{\Jin}\int_{\Jjm}\kernel(s,t)
			\drm s\drm t
			\\&=
			\int_{\Rbb_+\times \Rbb_+}
			\sum_{i=1}^{M_n}
			\sum_{j=1}^{M_m}
			\an_i\am_j
			\mathbf{1}_{\Jin}(t)
			\mathbf{1}_{\Jjm}(s)
			\kernel(s,t)
			\drm s\drm t.
		\end{split}
	\end{equation*}
	Therefore, due to the definition of  $\vcv_n$ and $\vcv_m$, we have   
	\begin{equation} 
		l_{n,m} = \inner{\vcf_n}{\vcf_m}_{\Hk} =
		\int_{\Rbb_+\times \Rbb_+}
		\vn_s\kernel(s,t)\vm_t
		\drm s\drm t.
	\end{equation}
	Accordingly, since $\|\vcv_n\|_{\infty},
	\|\vcv_m\|_{\infty}\le \|\vcv\|_{\infty}$, one can see that 
	\begin{equation}\label{eqn:|l_mn-l_n|} 
		\begin{split}\!\!\!\!\!\!
			|l_{n,m}-l_n| &\!=\! 
			\bigg| \int_{\Rbb_+\times\Rbb_+}\!\!
			\vn_t \Big[\kernel(t,s)(\vm_s\!-\!v_s)\Big]
			\drm t\drm s\Big|
			\\&\!\le\!
			\|\vcv\|_{\infty} 
			\int_{\Rbb_+\times\Rbb_+}\!\!
			\Big| \kernel(t,s)(\vm_s \!-\! v_s)\Big| 
			\drm t\drm s.
		\end{split}
	\end{equation}
	From $\|\vcv_n\|_{\infty},
	\|\vcv_m\|_{\infty}\le \|\vcv\|_{\infty}$,
	we have  
	\begin{equation} 
		|\kernel(t,s)(\vm_s-v_s)|
		\le 
		2\|\vcv\|_{\infty} |\kernel(t,s)|.
	\end{equation}
	Moreover, for almost all $s\in\Rbb_+$, we know that 
	\begin{equation} 
		\lim_{m\to\infty}
		|\kernel(t,s)(\vm_s-v_s)| = 0.
	\end{equation}
	Since $\kernel$ is integrable, from the dominated convergence theorem~\cite{stein2009real}, it follows that 
	\begin{equation} 
		\lim_{m\to\infty}
		\int_{\Rbb_+\times\Rbb_+}
		|\kernel(t,s)(\vm_s-v_s)|
		\drm s \drm t = 0.
	\end{equation}
	Therefore, due to \eqref{eqn:|l_mn-l_n|}, there exists $M_{\varepsilon}$ such that, for any $m\ge M_{\varepsilon}$, we have
	$|l_{n,m}-l_n|\le \frac18\varepsilon^2$.
	Accordingly, from triangle inequality, we have $|l_{m,n}-l|\le \frac14\varepsilon^2$, for any $m,n\ge K_{\varepsilon}:= \max\{M_{\varepsilon},N_{\varepsilon}\}$. 
	Subsequently, it follows that  
	\begin{equation} 
		\begin{split}
			\|\vcf_n&-\vcf_m\|^2 
			=
			\inner{\vcf_n}{\vcf_n}_{\Hk} -2\inner{\vcf_n}{\vcf_m}_{\Hk} +\inner{\vcf_m}{\vcf_m}_{\Hk} \!\!\!\!\!\!\!\!\!
			\\&= 
			l_{n,n}-2l_{n,m}+l_{m,m}
			\\&\le 
			(l+\frac14\varepsilon^2)-2(l-\frac14\varepsilon^2)+(l+\frac14\varepsilon^2) = \varepsilon^2.
		\end{split}
	\end{equation}
	Hence, for any $m,n\ge K_{\varepsilon}$, we have $\|\vcf_n-\vcf_m\|_{\Hk} \le \varepsilon$.
	Therefore, $\{\vcf_n\}_{n=1}^{\infty}$ is a Cauchy sequence in $\Hk$, and there exists $\vcf=(f_s)_{s\in\Rbb_+}\in\Hk$ such that $\lim_{n\to\infty}\vcf_n=\vcf$.
	Accordingly, due to the reproducing property, for any $t\in\Rbb_+$, we have  
	\begin{equation}\label{eqn:ft_lim_fnt} 
		\lim_{n\to\infty}f_{n,t}=
		\lim_{n\to\infty}\inner{\kernel_t}{\vcf_n}_{\Hk}=
		\inner{\kernel_t}{\vcf}_{\Hk}=f_t.
	\end{equation}
	For any $n\in\Nbb$ and for almost all $s,t\in\Rbb_+$, we have  
	\begin{equation} 
		\Big| \kernel(t,s)\vn_s \Big|
		\le \|\vcv\|_{\infty} |\kernel(t,s)|,
	\end{equation}
	and 
	\begin{equation} 
		\lim_{n\to\infty} \kernel(t,s)\vn_s =
		\kernel(t,s) v_s.
	\end{equation}
	Accordingly, since $\kernel$ is integrable, from the dominated convergence theorem~\cite{stein2009real}, \eqref{eqn:f_n} and \eqref{eqn:ft_lim_fnt}, it follows that 
	\begin{equation}\label{eqn:ft_intkts} 
		\begin{split}
			f_t 
			&=
			\lim_{n\to\infty}
			\int_{\Rbb_+}\kernel(t,s)\vn_s\drm s
			\\&=
			\int_{\Rbb_+}\lim_{n\to\infty}\kernel(t,s)\vn_s\drm s
			\\&= 
			\int_{\Rbb_+}\kernel(t,s)v_s\drm s,
		\end{split}
	\end{equation} 
	i.e., we have 
	\begin{equation}\label{eqn:f_intk} 
		\vcf=\int_{\Rbb_+}\kernel(\cdot,s)v_s\drm s.	
	\end{equation}
	For almost all $s\in\Rbb_+$, we know that
	$\lim_{n\to\infty}g_s\vn_s = g_sv_s$.
	Moreover, one has that
	$|g_s\vn_s|\le \|\vcv\|_{\infty}|g_s|$, for each $n\in\Nbb$.
	Since $\vcg=(g_s)_{s\in\Rbb_+}\in\Hk$ and each element of $\Hk$ is integrable, 
	due to the dominated convergence theorem~\cite{stein2009real}, \eqref{eqn:int_gs_vns_inner_fn_g} and $\lim_{n\to\infty}\vcf_n=\vcf$,
	we have 
	\begin{equation} 
		\label{eqn:int_gs_vs_inner_fg}
		\begin{split}
			\int_{\Rbb_+}g_sv_s\drm s
			&
			=
			\lim_{n\to\infty}
			\int_{\Rbb_+}g_s\vn_s\drm s
			\\&
			=
			\lim_{n\to\infty}
			\inner{\vcf_n}{\vcg}_{\Hk} 
			\\&
			=
			\inner{\vcf}{\vcg}_{\Hk}. 
		\end{split}
	\end{equation}  
	Let $\phiu{\tau}=(\phiu{\tau,t})_{t\in\Rbb_+}$
	be defined such that for any $t\in\Rbb_+$, 
	we have  
	\begin{equation} 
		\phiu{\tau,t} =  \int_{\Rbb_+}\kernel(t,s)u_{\tau-s}\drm s,
	\end{equation}
	i.e., $\phiu{\tau}=\int_{\Rbb_+}\kernel(\cdot,s)u_{\tau-s}\drm s$.
	Due to \eqref{eqn:ft_intkts} and the fact that $v_s=u_{\tau-s}$, for $s\in\Rbb_+$, we know that $\phiu{\tau}=\vcf\in\Hk$.
	Accordingly, from \eqref{eqn:int_gs_vs_inner_fg}, we have 
	\begin{equation*} 
	\begin{split}
		\Lu{\tau}(\vcg)
		&=
		\int_{\Rbb_+}
		v_s u_{\tau-s} \drm s
		\\&=
		\Big\langle \int_{\Rbb_+}
			\kernel(\cdot,s)u_{\tau-s}\drm s,\vcg\Big\rangle_{\Hk} 
		\\ &=
		\inner{\phiu{\tau}}{\vcg}_{\Hk},
	\end{split}
	\end{equation*}
	which implies that $\Lu{\tau}$ is a continuous (bounded) operator on $\Hk$. This concludes the proof for the case of $\Tbb=\Rbb_+$. 
	
	\textbf{Case II:} Let $\Tbb=\Zbb_+$ and, similarly to the previous case, define $\vcv=(v_s)_{s\in\Zbb_+}$ as $v_s=u_{\tau-s}$, for any $s\in\Zbb_+$. One can easily see that $\|\vcv\|_{\infty}=\|\vcu\|_{\infty}$. 
	For any $\vcg=(g_s)_{s\in\Zbb_+}\in\Hk$, we know that $\sum_{s\in\Zbb_+}|g_s|<\infty$, which implies that $\Lu{\tau}(\vcg)=\sum_{s\in\Zbb_+}g_sv_s$ is absolutely convergent due to $\|\vcv\|_{\infty}=\|\vcu\|_{\infty}<\infty$.
	Let $\varepsilon$ be an arbitrary positive real scalar.
	Since $\kernel$ is summable, there exists $N_{\varepsilon}\in\Nbb$ such that  
	\begin{equation} 
		\sum_{s,t\ge N_{\varepsilon}+1}
		|\kernel(t,s)|\le \frac{1}{\|\vcv\|_{\infty}^2}
		\varepsilon^2.  
	\end{equation}
	For any $n\in\Nbb$, let  $\vcf_n=(f_{n,t})_{t\in\Zbb_+}$ be defined as
	\begin{equation}
		f_{n,t}=\sum_{s=0}^n\kernel(t,s)v_s,
		\qquad \forall t\in \Zbb_+. 	
	\end{equation}
	One can see that $\vcf_n\in\Hk$. Let $n,m\in\Nbb$ such that $n,m\ge N_{\varepsilon}$.
	Without loss of generality, assume  $n\ge m$. Due to the reproducing property, we have 
	\begin{equation} 
		\begin{split}
			\|\vcf_n-\vcf_m\|_{\Hk}^2
			& =
			\Big\|\sum_{s=n+1}^m\kernel(\cdot,s)v_s\Big\|_{\Hk}^2
			\\&=
			\sum_{s,t\ge n+1}^m\kernel(s,t)v_s v_t
			\\&\le 
			\|\vcv\|_{\infty}^2\sum_{s,t\ge N_{\varepsilon}+1}|\kernel(s,t)|
			\\&\le 
			\varepsilon^2,
		\end{split}
	\end{equation}
	i.e., $\|\vcf_n-\vcf_m\|_{\Hk}\le \varepsilon$.
	Therefore $\{\vcf_n\}_{n=1}^{\infty}$ is a Cauchy sequence in $\Hk$, and there exists $\vcf=(f_t)_{t\in\Zbb_+}$ such that $\lim_{n\to\infty}\vcf_n=\vcf$.
	Note that, we have  
	\begin{equation} 
		f_t=
		\inner{\kernel_t}{\vcf}_{\Hk}=
		\lim_{n\to\infty}
		\inner{\kernel_t}{\vcf_n}_{\Hk}=
		\lim_{n\to\infty}f_{n,t},
	\end{equation}
	for any $t\in\Zbb_+$.
	Accordingly, from the reproducing property, one can see that 
	\begin{equation} 
		\begin{split}
			f_t
			&=
			\lim_{n\to\infty}
			\Big\langle \kernel_t,\sum_{s=0}^n\kernel(\cdot,s)v_s\Big\rangle_{\Hk} 
			\\&=
			\lim_{n\to\infty}
			\sum_{s=0}^n\kernel(t,s)v_s 
			\\&= 
			\sum_{s=0}^{\infty}\kernel(t,s)v_s,
		\end{split}
	\end{equation}
	where the last equality is due to  $\sum_{s=0}^{\infty}|\kernel(t,s)v_s|<\infty$, for any $t\in\Zbb_+$.
	Hence, we have $\vcf=\sum_{s=0}^{\infty}\kernel(\cdot,s)v_s$.
	For any $\vcg=(g_s)_{s\in\Zbb_+}\in\Hk$, we know that $\sum_{s\in\Zbb_+}|g_s|\le \infty$, which implies that $\Lu{\tau}(\vcg)=\sum_{s\in\Zbb_+}g_sv_s$ is absolutely convergent due to $\|\vcv\|_{\infty}=\|\vcu\|_{\infty}<\infty$.
	Therefore, one can see that 
	\begin{equation} 
		\begin{split}
			\sum_{s\in\Zbb_+}g_sv_s
			&=
			\lim_{n\to\infty}
			\sum_{s=0}^nv_s\inner{\kernel_s}{\vcg}_{\Hk}
			\\&=
			\lim_{n\to\infty}
			\Big\langle \sum_{s=0}^n\kernel(\cdot,s)v_s,\vcg\Big\rangle_{\Hk} 
			\\&=
			\lim_{n\to\infty}
			\inner{\vcf_n}{\vcg}_{\Hk} \\ &= \inner{\vcf}{\vcg}_{\Hk}.
		\end{split}
	\end{equation}
	Let $\phiu{\tau}=(\phiu{\tau,t})_{t\in\Zbb_+}$ be defined as 
	\begin{equation} 
		\phiu{\tau,t}=f_t=\sum_{s\in\Zbb_+}\kernel(t,s)v_s,	
	\end{equation}
	for any $t\in\Zbb_+$.
	Accordingly, we have 
	\begin{equation*} 
		\Lu{\tau}(\vcg) 
		\!=\! \sum_{s\in\Zbb_+}g_sv_s
		\!=\!
		\langle\phiu{\tau},\vcg\rangle_{\Hk} 
		\!=\! 
		\Big\langle\sum_{s\in\Zbb_+}\kernel(\cdot,s)v_s,\vcg\Big\rangle_{\Hk}, 
	\end{equation*}
	for any $\vcg\in\Hk$.
	This concludes the proof.
\end{proof} 
From this theorem, we have the following corollary. 
\begin{corollary}\label{thm:KRI_existence_uniquenes_solution}
	Let $\kernel$ be an integrable kernel and $\vcu\in\Lscrinfty$. 
	Then, the kernel-based impulse response estimation problem \eqref{eqn:kernel_based_IR_identification} admits a unique solution introduced in \eqref{eqn:solution_KRI}.
\end{corollary} 
\begin{proof}
From Theorem~\ref{thm:Lu_bounded}, it follows that the objective in \eqref{eqn:kernel_based_IR_identification} is function $\Jcal:\Hk\to\Rbb$ defined as  
\begin{equation} 
\Jcal(\vcg)=\sum_{i=1}^{\nD}\big(\inner{\phiu{t_i}}{\vcg}_{\Hk}-y_{t_i}\big)^2 + \lambda \|\vcg\|_{\Hk}^2,
\end{equation}	
for any $\vcg\in\Hk$.
This implies that  $\Jcal$ is a quadratic continuous function. Since $\lambda>0$, we know that $\Jcal$ is strongly convex.
Accordingly, from  $\Jcal(\zero)=\|\vcy\|^2$, it follows that $\Jcal$ is a proper continuous strongly convex function.
Therefore, due to \cite[Theorem 2.19]{peypouquet2015convex}, we know that $\min_{\vcg\in\Hk}\Jcal(\vcg)$ has a unique solution, which implies the existence and uniqueness for the solution of \eqref{eqn:kernel_based_IR_identification}.
The proof concludes from \eqref{eqn:phiu_tau_t} and \cite[Theorem 1.3.1]{wahba1990spline}.
\end{proof}

\section{Conclusion}\label{sec:conclusion}  
The kernel-based system identification stands on the central assumption that the convolution operators restricted to the chosen RKHS are continuous linear functionals. Current research work  in the literature assumes, implicitly or explicitly, that this continuity property holds without elaborating the required conditions. In this work, we have addressed this long-standing question by specifying these conditions: the boundedness of the input signal and the integrability of the kernel function. Furthermore, owing to the strong convexity of the optimization problem and the resulted continuity feature of the convolution operators, we have shown that the kernel-based approach is well-defined by guaranteeing the existence and uniqueness properties for the solution of the identification problem.

\appendix
\section{Appendix}   
\label{sec:Appendix}
\subsection{Proof of Lemma \ref{lem:int_kernel_I}}\label{apn:proof:lem:int_kernel_I}  
First we show the claims when $\otau<\infty$, and then, extend the result to the general case.  

Let $\Dtau$ be defined as $\Dtau=\otau-\utau$. For each $n\in\Nbb$, define $N_n$ and $\Delta_n$  respectively as $N_n:=2^n$ and $\Delta_n:=2^{-n}\Dtau$. Also, let function $\vcf_n:=(f_{n,s})_{s\in\Rbb_+}\in\Hk$ be defined as  
\begin{equation}\label{eqn:def_fn}  
	f_{n,s}=\Delta_n\sum_{i=1}^{N_n}\kernel(\tin,s), \qquad \forall s\in\Rbb_+,
\end{equation}
where $\tin = \utau+(i-1)\Dtau$, for  $i=1,\ldots,2^n$.
Let $\varepsilon$ be an arbitrary positive real scalar.
Since $\kernel$ is continuous, we know that it is uniformly continuous on compact region $[\utau,\otau]\times[\utau,\otau]$. Therefore, there exists positive real scalar $\delta$ such that for any $(s_1,t_1),(s_2,t_2)\in [\utau,\otau]\times[\utau,\otau]$ where $|s_1-s_2|+|t_1-t_2|\le \delta$, we have   
\begin{equation}\label{eqn:k_s1t1_k_s2t2}  
	|\kernel(s_1,t_1)-\kernel(s_2,t_2)|\le \frac{1}{4}\frac{\varepsilon^2}{\Dtau^2}.
\end{equation} 
Let ${\delta}_{\varepsilon}$ be the largest scalar in $(0,1)$ with such property and define 
$N_{\varepsilon}$ as the smallest integer such that  
\begin{equation}  
	N_{\varepsilon}\ge \max(-\log_2(\frac{{\delta}_{\varepsilon}}{\Dtau}),0)+1.
\end{equation}
Consider arbitrary integers $n$ and $m$ such that $n,m\ge N_{\varepsilon}$.
From the reproducing property of the kernel,  one can see that  
\begin{equation}\label{eqn:inner_fn_fm}  
	\begin{split}\!\!\!\!\!\!
		\inner{\vcf_n}{\vcf_m}_{\Hk} &\!=\!
		\Big\langle\Delta_n\sum_{i=1}^{N_n}\kernel(\tin,\cdot),
		\Delta_m\sum_{j=1}^{N_m}\kernel(\tjm,\cdot)\Big\rangle_{\Hk}\!\!\!\!\!\!\!\!\!
		\\&\!=\!
		\Delta_n\Delta_m\sum_{i=1}^{N_n}\sum_{j=1}^{N_m}\kernel(\tin,\tjm).
	\end{split}
\end{equation}
Define $\Iijnm$ as region $[\tin,\tin+\Delta_n)\times[\tjm,\tjm+\Delta_m)$, for $i=1,\ldots,2^n$ and $j=1,\ldots,2^m$. Also, let $I$ be the value defined as   
\begin{equation}  
	I:=\int_{[\utau,\otau]\times[\utau,\otau]}\kernel(s,t)\drm s\drm t.
\end{equation}
Note that $I$ is a well-defined integral due to integrability of $\kernel$.
From \eqref{eqn:inner_fn_fm} and the triangle inequality, we have  
\begin{equation*}  
	\begin{split}
		|\langle\vcf_n,&\vcf_m\rangle_{\Hk}\!-\! I| 
		\\&=
		\bigg|\!\sum_{i=1}^{N_n}\sum_{j=1}^{N_m}
		\int_{\Iijnm}
		\!\! \kernel(\tin,\tjm)-\kernel(s,t)\drm s \drm t
		\bigg|
		\\ &\le 
		\sum_{i=1}^{N_n}\sum_{j=1}^{N_m}
		\int_{\Iijnm}
		\big|\kernel(\tin,\tjm)-\kernel(s,t)\big|
		\drm s \drm t
		\\ &\le 
		\sum_{i=1}^{N_n}\sum_{j=1}^{N_m}
		\frac{1}{4}\frac{\varepsilon^2}{\Dtau^2} \Delta_n\Delta_m
		\\ &= 
		\frac{1}{4}\varepsilon^2,
	\end{split}
\end{equation*}
where the inequality is due to \eqref{eqn:k_s1t1_k_s2t2}.
Subsequently, one can see that  
\begin{equation}  
	\label{eqn:I_inner_fn_fm_eps}
	I-\frac{1}{4}\varepsilon^2
	\le
	\inner{\vcf_n}{\vcf_m}
	\le
	I + \frac{1}{4}\varepsilon^2.
\end{equation} 
From \eqref{eqn:I_inner_fn_fm_eps}, it follows that  
\begin{equation*}  
	\begin{split}
		\|\vcf_n-\vcf_m&\|_{\Hk}^2 
		=
		\inner{\vcf_n}{\vcf_n}_{\Hk}
		-2\inner{\vcf_n}{\vcf_m}_{\Hk} + \inner{\vcf_m}{\vcf_m}_{\Hk}
		\\&\le
		(I + \frac{1}{4}\varepsilon^2)
		-2(I-\frac{1}{4}\varepsilon^2) +
		(I + \frac{1}{4} \varepsilon^2) 
		\\ &= 
		\varepsilon^2,
	\end{split}
\end{equation*} 
and, hence, we have $\|\vcf_n-\vcf_m\|_{\Hk}\le \varepsilon$.
Therefore,  $\{\vcf_n\}_{n=1}^{\infty}$ is a Cauchy sequence in $\Hk$, and there exists $\vcf=(f_s)_{s\in\Rbb_+}\in\Hk$ such that $\lim_{n\to\infty}\|\vcf_n-\vcf\|_{\Hk}=0$. 
For any $s\in\Rbb_+$, due to the Cauchy-Schwartz inequality and the reproducing property, we have  
\begin{equation}\label{eqn:fn_s_f_s_le} 
	|f_{n,s}-f_s| = |\inner{\kernel_s}{\vcf_n-\vcf}_{\Hk}|
	\le \kernel(s,s)^{\frac12}\|\vcf_n-\vcf\|_{\Hk},
\end{equation}
which implies that 
$\lim_{n\to\infty}f_{n,s}=f_s$.
On the other hand, from \eqref{eqn:k_s1t1_k_s2t2}, one can see that  
\begin{equation*}  
	\begin{split}
		\bigg|f_{n,s}\!\!\ -\int_{[\utau,\otau]}\!\!\kernel(s,t)\drm t\bigg| 
		&=
		\bigg|\Delta_n\sum_{i=1}^{N_n}\kernel(\tin,s)
		-\! \int_{\utau}^{\otau}\!\!\kernel(s,t)\drm t
		\bigg|
		\\&=
		\bigg|\sum_{i=1}^{N_n}
		\int_{\Iin}
		\!\! \kernel(\tin,s)-\kernel(s,t) \drm t
		\bigg|
		\\ &\le 
		\sum_{i=1}^{N_n}
		\int_{\Iin}\big|
		\kernel(\tin,s)-\kernel(s,t)\big| \drm t
		\\&\le 
		\sum_{i=1}^{N_n}
		\frac{1}{4}\frac{\varepsilon^2}{\Dtau^2} \Delta_n
		\\ &= 
		\frac{1}{4}\frac{\varepsilon^2}{\Dtau} ,
	\end{split}
\end{equation*}
where, for $i=1,\ldots,N_n$, interval $\Iin$ is defined as $[\tin,\tin+\Delta_n)$.
Accordingly, we have  
\begin{equation}  
	f_s=\limOp_{n\to\infty}f_{n,s} = \int_{[\utau,\otau]}\kernel(s,t)\drm t,
\end{equation}
which says that $\vcf=\int_{[\utau,\otau]}\kernel(\cdot,t)\drm t \in \Hk$.
Moreover, from $\lim_{n\to\infty}\vcf_n=\vcf$, we know that  
\begin{equation}  
	\|\vcf\|_{\Hk}^2
	=
	\limOp_{n\to\infty}\|\vcf_n\|_{\Hk}^2
	=
	\limOp_{n\to\infty}\inner{\vcf_n}{\vcf_n}_{\Hk}.
\end{equation} 
Therefore, from \eqref{eqn:I_inner_fn_fm_eps} and the definition of $\vcf$ and $I$, it follows that  
\begin{equation}  
	\begin{split}
		\Big\|\int_{[\utau,\otau]}\kernel(\cdot,t)\drm t\Big\|_{\Hk}^2
		&= \limOp_{n\to\infty}\inner{\vcf_n}{\vcf_n}_{\Hk}
		\\&
		= \int_{[\utau,\otau]\times[\utau,\otau]}\kernel(s,t)\drm s\drm t.
	\end{split}
\end{equation}
Note that, for any $s$ such that $t+s\in\Rbb_+$, due to the reproducing property of the kernel, 
we have  
\begin{equation*}  
	\begin{split}
		\|\kernel_{t+s}-\kernel_t\|_{\Hk}^2
		=&\ \kernel(t+s,t+s)\\&-\kernel(t,t+s)-\kernel(t+s,t)+\kernel(t,t),
	\end{split}
\end{equation*}
which implies that $\lim_{s\to 0}\|\kernel_{t+s}-\kernel_t\|_{\Hk}=0$ due to continuity of the kernel.
Meanwhile, for each $\vcg=(g_t)_{t\in\Rbb_+}\in\Hk$, 
from the Cauchy-Schwartz inequality and the reproducing property,
one has  
\begin{equation*}  
\begin{split}
	|g_{t+s}-g_t| 
	&= 
	|\inner{\kernel_{t+s}-\kernel_t}{\vcg}_{\Hk}| 
	\\&\le 
	\|\kernel_{t+s}-\kernel_t\|_{\Hk} \|\vcg\|_{\Hk}.
\end{split}
\end{equation*}
Accordingly, we have $\lim_{s\to 0}|g_{t+s}-g_t|=0$, which says that $\vcg=(g_t)_{t\in\Rbb_+}$ is a continuous function of $t$.
Hence, the Riemann integral of $\vcg$ exists, and we have  
\begin{equation*}  
	\begin{split}
		\int_{[\utau,\otau]}  g_t\drm t 
		&= 
		\limOp_{n\to\infty}\sum_{i=1}^{2^n}g\big(\utau+(i-1)2^{-n}\Dtau\big)2^{-n}\Dtau
		\\
		&= 
		\limOp_{n\to\infty}\sum_{i=1}^{N_n}g(\tin)\Delta_n 
		\\
		&=  \limOp_{n\to\infty}\inner{\vcf_n}{\vcg}_{\Hk},
	\end{split}
\end{equation*}
where the last equality is due the definition of $\vcf_n$ in \eqref{eqn:def_fn} and the reproducing property.
Therefore, from $\lim_{n\to\infty}\vcf_n = \vcf$, one can see that  
\begin{equation}  
	\begin{split}
		\int_{[\utau,\otau]}  g_t\drm t &=
		\limOp_{n\to\infty}\inner{\vcf_n}{\vcg}_{\Hk}
		\\&=
		\inner{\vcf}{\vcg}_{\Hk}
		\\&= \Big\langle\int_{[\utau,\otau]}\kernel(\cdot,t)\drm t,\vcg\Big\rangle_{\Hk}.
	\end{split}
\end{equation}
Now, we consider the case where $\otau=\infty$. 
For each integer $n\ge \tau$, let  
$\vch_n =(h_{n,s})_{s\in\Rbb_+}$ be function
$\vch_n=\int_{[\utau,n]}\kernel(\cdot,t)\drm t$ which is well-defined and belongs to $\Hk$.
Let $\varepsilon$ be an arbitrary positive real scalar.
Since $\kernel$ is absolutely integrable, we know that   
\begin{equation}  
	\lim_{\tau\to\infty}\int_{\tau}^\infty\int_{\tau}^\infty|\kernel(s,t)|\drm s \drm t = 0.
\end{equation}
Let $\tau_{\varepsilon}\ge \utau$ be the smallest positive real scalar such that   
\begin{equation}\label{eqn:int_kernel_tau_eps}  
	\int_{\tau_{\varepsilon}}^\infty\int_{\tau_{\varepsilon}}^\infty|\kernel(s,t)|\drm s \drm t \le \varepsilon^2,
\end{equation}
and $n,m\in\Nbb$ be arbitrary indices such that $n,m\ge \tau_{\varepsilon}$.
Without loss of generality, assume $n\ge m$.
Then, due to the discussion above and the triangle inequality, we have  
\begin{equation}  
	\begin{split}\!\!
		\|\vch_n&-\vch_m\|_{\Hk}^2 = 
		\Big\|
		\int_m^n
		\kernel(\cdot,t)\drm t
		\Big\|_{\Hk}^2	
		\\&=
		\int_m^n
		\int_m^n
		\kernel(s,t)\drm s\drm t
		\\&\le 
		\int_m^n
		\int_m^n
		|\kernel(s,t)|\drm s\drm t
		\\&\le 
		\int_{\tau_{\varepsilon}}^\infty
		\int_{\tau_{\varepsilon}}^\infty
		|\kernel(s,t)|\drm s\drm t.\!\!\!\!\!\!\!\!\!\!\!\!
	\end{split}
\end{equation}  
Accordingly, from \eqref{eqn:int_kernel_tau_eps}, we know that $\|\vch_n-\vch_m\|_{\Hk}\le \varepsilon$, which implies that $\{\vch_n\}_{n\in\Nbb,n\ge \utau}$ is a Cauchy sequence in $\Hk$.  Therefore, there exists $\vch=(h_s)_{s\in\Rbb_+}\in\Hk$ such that $\lim_{n\to\infty}\vch_n=\vch$.
Based on an argument similar to \eqref{eqn:fn_s_f_s_le}, one can show that $\lim_{n\to\infty}h_{n,s}=h_s$, for any $s\in\Rbb_+$.  
Since $\kernel(s,\cdot)\in\Hk$ and the elements of $\Hk$ are integrable, due to the dominated convergence theorem~\cite{stein2009real}, we have  
\begin{equation}  
	\begin{split}
		h_s 
		& = \lim_{n\to\infty}h_n(s)
		\\&= \lim_{n\to\infty}\int_{\utau}^n\kernel(s,t)\drm t
		\\
		& = \lim_{n\to\infty}\int_{\utau}^{\infty}\kernel(s,t)\mathbf{1}_{[{\utau},n]}(t)\drm t
		\\&= \int_{\utau}^\infty\kernel(s,t)\drm t.\!\!\!\!\!\!
	\end{split}
\end{equation}  
In other words, one has $\vch =\int_{\utau}^\infty\kernel(\cdot,t)\drm t$.
Hence, from $\lim_{n\to\infty}\vch_n=\vch$ and the above discussion, we have  
\begin{equation}  
	\begin{split}\!\!\!\!
		\Big\| \!&\int_{\utau}^\infty\!\kernel(\cdot,t)\drm t\Big\|_{\Hk}^2 
		= \lim_{n\to\infty}\|\vch_n\|_{\Hk}^2
		\\& 
		= \!\lim_{n\to\infty}\int_{\utau}^n \!\!\!\int_{\utau}^n \! \kernel(s,t)\drm s\drm t
		\\& 
		= \!\lim_{n\to\infty}\int_{\utau}^\infty \!\!\!\int_{\utau}^\infty \!\! \kernel(s,t)\mathbf{1}_{[{\utau},n]^2}(s,t)\drm s\drm t
		\\&= \!\int_{\utau}^\infty\!\!\!\!\int_{\utau}^\infty \!\!\!\!\kernel(s,t)\drm s\drm t,\!\!\!\!\!\!\!\!\!\!\!\!
	\end{split}
\end{equation}  
where the last equality is according to the dominated convergence theorem.
Let $\vcg=(g_t)_{t\in\Rbb_+}$ be an arbitrary element of $\Hk$.
Based on same arguments as before, one can see that  
\begin{equation}  
	\begin{split}
		\Big\langle\int_{\utau}^\infty\kernel(\cdot,t)\drm t,&\vcg\Big\rangle_{\Hk}
		=
		\lim_{n\to\infty}\inner{\vch_n}{\vcg}_{\Hk}
		\\=& \lim_{n\to\infty} \Big\langle \int_{\utau}^n\kernel(\cdot,t)\drm t,\vcg\Big\rangle_{\Hk}
		\\=& \lim_{n\to\infty} \int_{\utau}^ng_t\drm t
		\\=& \lim_{n\to\infty} \int_{\utau}^\infty g_t\mathbf{1}_{[{\utau},n]}(t)\drm t 
		\\=& \int_{\utau}^\infty g_t\drm t,\!\!\!\!\!\!\!\!\!\!\!\!
	\end{split}
\end{equation}  
where the last equality is due to the dominated convergence theorem and the fact that $\vcg$ is integrable.   
\qed  
\subsection{Proof of Corollary \ref{cor:inner_int_I1_int_I2}}\label{apn:proof:cor:inner_int_I1_int_I2}  
In \eqref{eqn:int_g}, set $\utau$, $\otau$, and $\vcg$ respectively to $\utau_1$, $\otau_1$, $\vcg= \int_{[\utau_2,\otau_2]}
\kernel(\cdot,s)\drm s$. 
Accordingly, we have we 
\begin{equation*}
	\begin{split}
		\Big\langle\int_{\utau_1}^{\otau_1}\!\!\kernel(\cdot,t)\drm t,\!\int_{\utau_2}^{\otau_2}\!
		&\kernel(\cdot,s)\drm s\Big\rangle_{\Hk} 
		\\&
		=	
		\int_{\utau_1}^{\otau_1}
		\int_{\utau_2}^{\otau_2}
		\!\!\kernel(t,s)\drm s\drm t.
	\end{split}
\end{equation*}
In this equation, since the kernel is integrable, the right-hand side is well-defined. Moreover, due to the Fubini theorem~\cite{stein2009real}, it equals to the other integrals in \eqref{eqn:inner_int_I1_int_I2}. 
\qed

\subsection{Proof of Lemma \ref{lem:sum_kernel_I}}\label{apn:proof:lem:sum_kernel_I}
We know that $\kernel(\cdot,t)\in\Hk$, for each $t=\utau,\ldots,\otau$.  
If $\otau$ is finite, one can see that $\sum_{\utau\le t\le\otau}\kernel(\cdot,t)$ belongs to $\Hk$, hence, it is well-defined.
Moreover, using the definition of norm and the reproducing property, one can show \eqref{eqn:sum_kt}. Similarly, \eqref{eqn:sum_gt} is concluded from the reproducing property.
Now, we consider the case $\otau=\infty$.
For $n\ge \utau$, we define $\vcf_n\in\Hk$ as
$$\vcf_n=(f_{n,s})_{s\in\Zbb_+}:=\sum_{\utau\le t\le n}\kernel(\cdot,t).$$ 
Let $\varepsilon$ be an arbitrary positive real scalar, and  $N_{\varepsilon}$ be the smallest non-negative integer such that 
$$\sum_{s,t\ge N_{\varepsilon}}|\kernel(s,t)|\le \varepsilon.$$
Note that since $\kernel$ is integrable, there exist such $N_{\varepsilon}$ for any positive $\varepsilon$.
Now, let $n,m\in\Nbb$ such that $n,m\ge N_{\varepsilon}$ and without loss of generality, we assume $n\ge m$. 
Based on the previous case, we know that  
\begin{equation*}  
	\begin{split}
		\|\vcf_n-\vcf_m\|_{\Hk}^2 &= 
		\Big\|
		\sum_{t=m+1}^{n}\kernel(\cdot,t)
		\Big\|_{\Hk}^2
		\\&
		=
		\sum_{m+1\le s,t \le n}\kernel(s,t)
		\\&
		\le 
		\sum_{N_{\varepsilon}\le s,t }|\kernel(s,t)| \\&\le \varepsilon^2.
	\end{split}
\end{equation*}
Accordingly, we have  $\|\vcf_n-\vcf_m\|_{\Hk}\le \varepsilon$, which implies that $\{\vcf_n\}_{n\ge \utau}$ is a Cauchy sequence and hence convergent.
Let $\vcf=(f_s)_{s\in\Zbb_+}$ denote the limit of this sequence.
For any $s\in\Zbb_+$, we know that  
\begin{equation}  
	\begin{split}
	|f_s-f_{n,s}| 
	&=
	|\inner{\vcf-\vcf_n}{\kernel_s}_{\Hk}|
	\\&\le 
	\kernel(s,s)^{\frac12}\|\vcf-\vcf_n\|_{\Hk},
	\end{split}
\end{equation}
and consequently, we have $\lim_{n\to\infty}f_{n,s}=f_s$.
Subsequently, since $\kernel(s,\cdot)$ is absolutely integrable,  it follows that  
\begin{equation}  
	f_s = \lim_{n\to\infty}\sum_{\utau\le t\le n}\kernel(s,t) = \sum_{\utau\le t}\kernel(s,t),
\end{equation}
i.e., $\vcf=\sum_{\utau\le t}\kernel(\cdot,t)$.
Moreover, we have
\begin{equation}  
	\begin{split}\!\!  
		\Big\| \sum_{\utau\le t}\kernel(\cdot,t)\Big\|_{\Hk}^2 
		&= 
		\lim_{n\to\infty}\|\vcf_n\|_{\Hk}^2
		\\&= 
		\lim_{n\to\infty}\sum_{\utau\le s,t\le n}\kernel(s,t)
		\\&= 
		\sum_{\utau\le s,t}\kernel(s,t),
	\end{split}
\end{equation}  
where the last equality is due to the dominated convergence theorem~\cite{stein2009real} and being $\kernel$ integrable.
For any $\vcg=(g_t)_{t\in\Zbb_+}\in\Hk$, we know that $\vcg$ is integrable, i.e., $\sum_{t\in\Zbb_+}|g_t|<\infty$.
Therefore, 
from $\lim_{n\to\infty}\vcf_n = \vcf$, 
we have  
\begin{equation}  
	\begin{split}
		\sum_{\tau\le t}g_t 
		& 
		= \lim_{n\to\infty}\sum_{\tau\le t\le n}g_t 
		\\&
		= \lim_{n\to\infty}\Big\langle\sum_{\tau\le t\le n}\kernel(\cdot,t),\vcg\Big\rangle_{\Hk}
		\\&
		= \lim_{n\to\infty}\inner{\vcf_n}{\vcg}_{\Hk}
		\\&
		= \inner{\vcf}{\vcg}_{\Hk}
			\\&
		= \Big\langle\sum_{\utau\le t}\kernel(\cdot,t),\vcg\Big\rangle_{\Hk}.
	\end{split}
\end{equation}
This concludes the proof.
\qed

\bibliographystyle{IEEEtran}        
{\footnotesize{\bibliography{mainbib}}}           
\end{document}